\documentclass[envcountsame]{llncs}

\usepackage{latexsym}
\usepackage{verbatim}
\usepackage{amsmath,amsfonts,amssymb,mathrsfs,stmaryrd}
\usepackage{enumerate,times}

\newcommand{\bin}{\mathsf{bin}}
\newcommand{\val}{\mathsf{val}}
\newcommand{\height}{\mathsf{height}}

\newcommand{\NN}{\mathbb{N}}
\newcommand{\ZZ}{\mathbb{Z}}

\newcommand{\coNP}{\mathsf{coNP}}
 
\newcommand{\DSPACE}{\mathsf{DSPACE}}
\newcommand{\FSPACE}{\mathsf{FSPACE}}
\newcommand{\FTIME}{\mathsf{FTIME}}
\newcommand{\FPSPACE}{\mathsf{FPSPACE}}
\newcommand{\FpolyL}{\mathsf{FpolyL}}
\newcommand{\FP}{\mathsf{FP}}
\newcommand{\MTDD}{\mathrm{MTDD}}
\newcommand{\NC}{\mathsf{NC}}
\newcommand{\NSPACE}{\mathsf{NSPACE}}
\newcommand{\NP}{\mathsf{NP}}
\newcommand{\Ptime}{\mathsf{P}}
\newcommand{\polyL}{\mathsf{polyL}}
\newcommand{\PSPACE}{\mathsf{PSPACE}}

\newcommand{\ignore}[1]{{}}

\newenvironment{xmpl}{\begin{example}}{\end{example}}

\begin{document}

  \title{Processing Succinct Matrices and Vectors\thanks{The first (second) author is supported by the DFG grant LO 748/8-2 (SCHM 986/9-2).}}  %
  \titlerunning{Succinct Matrices}
\author{Markus Lohrey\inst{1}
 \and Manfred  Schmidt-Schau{\ss}\inst{2}
\institute{Universit\"at Siegen, Department f\"ur Elektrotechnik und Informatik, Germany
\and Institut f{\"u}r Informatik,  Goethe-Universit{\"a}t,  D-60054 Frankfurt, Germany}}
 
\maketitle

\begin{abstract}
 We study the complexity of algorithmic problems for matrices that are represented
 by multi-terminal decision diagrams (MTDD). These are a variant of ordered decision 
 diagrams, where the terminal nodes are labeled with arbitrary elements of a semiring  (instead of 
 $0$ and $1$).  A simple example shows that the product of two MTDD-represented
 matrices cannot be represented by an MTDD of polynomial size. To overcome this
 deficiency, we extended MTDDs to $\MTDD_+$ by allowing
 componentwise symbolic addition of variables (of the same dimension) in rules.
 It is shown that 
 accessing an entry,  equality checking, 
 matrix multiplication, and other basic matrix operations can be solved in
 polynomial time for $\MTDD_+$-represented matrices.
 On the other hand, testing whether the determinant of a
 MTDD-represented matrix vanishes 
 is $\PSPACE$-complete, and the same problem is
 $\NP$-complete for $\MTDD_+$-represented diagonal matrices. Computing a specific entry in a product of
 MTDD-represented matrices is $\#\mathsf{P}$-complete.
 \end{abstract}

\section{Introduction}

Algorithms that work on a succinct representation of certain objects 
can nowadays be found in many areas of computer science. A paradigmatic 
example is the use of OBDDs (ordered binary decision diagrams) in hardware 
verification  \cite{Bryant86,MeinelT98}. OBDDs are a succinct representation of Boolean functions.
Consider a boolean function $f(x_1,\ldots,x_n)$ in $n$ input variables.
One can represent $f$ by its decision tree, which is a full binary tree of height $n$
with $\{0,1\}$-labelled leaves. The leaf that is reached from the root via the path
$(a_1, \ldots, a_n) \in \{0,1\}^n$ (where $a_i=0$ means that we descend to the 
left child in the $i$-th step, and $a_i=1$ means that we descend to the 
right child in the $i$-th step) is labelled with the bit $f(a_1,\ldots,a_n)$.
This decision tree can be folded into a directed acyclic graph by eliminating 
repeated occurrences of isomorphic subtrees. The result is the OBDD for $f$ 
with respect to the variable ordering $x_1, \ldots, x_n$.\footnote{Here, we are cheating a bit:
In OBDDs a second elimination rule is applied that removes nodes for which the left and right
child are identical. On the other hand, it is known that asymptotically the compression achieved by this elimination rule
is negligible \cite{Wegener94a}.} Bryant was the first who realized that OBDDs are an adequate tool in order
to handle the state explosion problem in hardware verification \cite{Bryant86}.

OBDDs can be also used for storing large graphs. A graph $G$
with $2^n$ nodes and adjacency matrix $M_G$
can be represented by 
the boolean function $f_G(x_1,y_1, \ldots, x_n,y_n)$,
where $f_G(a_1,b_1, \ldots, a_n,b_n)$ is the entry of $M_G$ 
at position $(a,b)$; here $a_1 \cdots a_n$ (resp., $b_1 \cdots b_n$) is the binary representation 
of the index $a$ (resp. $b$). Note that we use the so called interleaved variable ordering
here, where the bits of the two coordinates $a$ and $b$ are bitwise interleaved. This
ordering turned out to be convenient in the context of OBDD-based graph representation,
see e.g. \cite{FujiiOH93}.

Classical graph problems (like reachability, alternating
reachability, existence of a Hamiltonian cycle) have been studied
for OBDD-represented graphs in
\cite{feigenbaum-kannan-vardi-viswanathan:99,Vei98CC}.
It turned out that these problems are exponentially harder 
for OBDD-represented graphs than for explicitly given graphs.
In \cite{Vei98CC} an upgrading theorem for OBDD-rep\-resented graphs was
shown. It roughly
states that completeness of a problem $A$ for a complexity class $C$
under quantifier free reductions implies
completeness of the OBDD-variant of $A$ for the exponentially harder 
version of $C$ under polynomial time reductions.

In the same way as OBDDs represent boolean mappings, 
functions from $\{0,1\}^n$ to any set $S$ can be represented. One simply has to label the leaves
of the decision tree with elements from $S$. This yields multi-terminal decision diagrams (MTDDs) \cite{FujitaMY97}. Of particular
interest is the case, where  $S$ is a semiring, e.g. $\NN$ or $\ZZ$.
In the same way as an adjacency matrix (i.e., a boolean matrix) of dimension $2^n$
can be represented by an OBDD,  a matrix of dimension $2^n$ over any semiring can be represented by an MTDD.
As  for OBDDs, we assume that the bits of the two coordinates $a$ and $b$ are interleaved in the order
$a_1,b_1, \ldots, a_n,b_n$.  This 
implies that an MTDD can be viewed as a set of rules of the form
\begin{equation} \label{eq-rules-of-MTDD}
A \to \left(\begin{array}{ll} A_{1,1} &  A_{1,2}\\ A_{2,1} &
      A_{2,2}\end{array}\right)  
 \qquad   \text{ or } \qquad     
      B \to a \ \text{ with } \ a \in S.
\end{equation}
where $A$, $A_{1,1}$, $A_{1,2}$, $A_{2,1}$, and $A_{2,2}$ are 
variables that correspond to certain nodes of the MTDD (namely
those nodes that have even distance from the root node).
Every variable  produces a matrix of dimension $2^h$ for some 
$h \geq 0$, which we call the height of the variable.
The variables $A_{i,j}$ in \eqref{eq-rules-of-MTDD}
must have the same height $h$, and $A$ has height $h+1$. 
The variable $B$ has height $0$. We assume that the additive monoid
of the semiring $S$ is finitely generated, hence every $a \in S$ has a finite 
representation.

MTDDs yield very compact representations of sparse matrices. It was
shown that an $(n \times n)$-matrix with $m$ nonzero entries can be represented by an MTDD
of size $O(m \log n)$ \cite[Theorem~3.2]{FujitaMY97}, which is better than standard succinct representations for sparse matrices.
Moreover, MTDDs can also yield very compact representations of non-sparse matrices. For instance, the Walsh matrix
of dimension $2^n$ can be represented by an MTDD of size $O(n)$, see \cite{FujitaMY97}. In fact, the usual definition
of the $n$-th Walsh matrix is exactly an MTDD.
Matrix algorithms for MTDDs
are studied in \cite{FujitaMY97} as well, but no precise complexity
analysis is carried out.  In fact, the straightforward matrix
multiplication algorithm for multi-terminal decision diagrams from \cite{FujitaMY97}
has  an exponential worst case running time, and this is unavoidable:
The smallest MTDD that produces the product of two
MTDD-represented matrices may be of exponential size in the two 
MTDDs, see Theorem~\ref{prop-MTDD-mult}. 
The first main contribution of this paper is a  generalization of MTDDs
that overcomes this deficiency: An 
$\MTDD_+$ consists of rules of the form \eqref{eq-rules-of-MTDD} together 
with addition rules of the form 
$A \to B+C$, where ``$+$'' refers to 
matrix addition over the underlying semiring. Here, $A$, $B$, and $C$ must have the same height, i.e.,
produce matrices of the same dimension.  We show that an $\MTDD_+$ 
for the product of two $\MTDD_+$-represented matrices
can be computed in polynomial time (Theorem~\ref{thm:main-matrix-mult}). 
In Section~\ref{sec-MTDD}
we also present efficient (polynomial time)
algorithms for several other important matrix problems on 
$\MTDD_+$-represented input matrices: 
computation of a specific matrix entry, computation of the trace,
matrix transposition, tensor and Hadamard product.
Section~\ref{sec-eq} deals with equality checking. It turns 
out that equality of $\MTDD_+$-represented matrices can be checked
in polynomial time, if the additive monoid is cancellative, in all
other cases equality checking is {\sf coNP}-complete.

To the knowledge of the authors, complexity results similar to those from 
\cite{feigenbaum-kannan-vardi-viswanathan:99,Vei98CC} for OBDDs do not exist
in the literature on MTDDs.  Our second main contribution fills this gap.
We prove that already for MTDDs over $\ZZ$ it is $\PSPACE$-complete to
check whether the determinant of the generated matrix is zero (Theorem~\ref{theorem-det}).
This result is shown by lifting a classical construction of Toda \cite{Toda91countingproblems}
(showing that computing the determinant of an explicitly given 
integer matrix is complete for the counting class $\mathsf{GapL}$) 
to configuration graphs of polynomial space bounded Turing machines, which
are of exponential size. It turns out that the adjacency matrix of the configuration graph of a 
polynomial space bounded Turing machine can be produced by a small
MTDD. Theorem~\ref{theorem-det} sharpens a recent result from 
\cite{GreKoiPort12} stating that it is $\PSPACE$-complete to check whether the 
determinant of a matrix that is represented by a boolean circuit (see 
Section~\ref{sec-bool-circuits}) vanishes.
We also prove several hardness results for counting classes. For
instance, computing a specific entry of a matrix power $A^n$, where
$A$ is given by an MTDD over $\NN$
is $\#\mathsf{P}$-complete (resp. $\#\mathsf{PSPACE}$-complete) if $n$ is given
unary (resp. binary). Here, $\#\mathsf{P}$ (resp. $\#\mathsf{PSPACE}$)
is the class of
functions counting the number of accepting computations of a
nondeterministic polynomial time Turing machine \cite{Valiant79a}
(resp., a nondeterministic polynomial space Turing machine \cite{Ladner89}).
An example of a natural $\#\mathsf{PSPACE}$-complete counting problem  
is counting the number of strings not accepted by a given 
NFA \cite{Ladner89}.

\section{Related work}

\paragraph*{\bf Sparse matrices and quad-trees.}
To the knowledge of the authors, most of the literature on matrix 
compression deals with sparse matrices, where most of the matrix
entries are zero. There are several succinct representations of
sparse matrices. One of which are  {\em quad-trees}, used in computer graphics for the
representation of large constant areas in 2-dimensional pictures,
see for example \cite{samet:90,eppstein-quadtree-GS:08}.
Actually, an MTDD can be seen as a quad-tree that is folded into
a dag by merging identical subtrees. 

\paragraph*{\bf Two-dimensional straight-line programs.}

MTDDs are also a special case of 2-dimen\-sional
straight-line programs (SLPs).
A (1-dimensional) SLP is a context-free grammar in Chomsky normal form that generates
exactly one OBDD. An SLP with $n$ rules can generate a string
of length $2^n$; therefore an SLP can be seen as a succinct
representation of the string it generates.
Algorithmic problems that can be solved efficiently (in polynomial
time) on SLP-represented strings are for instance equality checking 
(first shown by Plandowski \cite{Plandowski:94}) and pattern matching, see 
\cite{lohrey-survey-13} for a survey.

In \cite{berman-karpinski-2d:02} a 2-dimensional extension of SLPs 
(2SLPs in the following) was defined. Here, every variable of the
grammar generates a (not necessarily square) matrix  (or picture), where every
position is labeled with an alphabet symbol. Moreover, there are two
(partial) concatenation operations: horizontal composition (which is defined
for two pictures if they have the same height) and vertical composition (which is defined
for two pictures if they have the same width).
This formalism does not share all the nice algorithmic
properties of (1-dimensional) SLPs \cite{berman-karpinski-2d:02}:
Testing whether two 2SLPs produce the same picture is only
known to be in $\mathsf{coRP}$ (co-randomized polynomial time).
Moreover,  checking whether an explicitly given (resp., 2SLP-represented) picture appears within 
a 2SLP-represented picture is $\NP$-complete (resp., $\Sigma_2^P$-complete).
Related hardness results in this direction concern the
convolution of two SLP-represented strings of the same length (which can be
seen as a picture of height 2). The convolution of strings
$u = a_1 \cdots a_n$ and $v = b_1\cdots b_n$ is 
the string $(a_1,b_1) \cdots (a_n,b_n)$.
By a result from \cite{bertoni-choffrut:08} (which is stated
in terms of the related operation of literal shuffle), the size of a
shortest  SLP for the convolution
of two strings that are given by SLPs $G$ and $H$
may be exponential in the size of $G$ and $H$. Moreover,
it is $\PSPACE$-complete to check for two SLP-represented
strings  $u$ and $v$ and an NFA $T$ 
operating on strings of pairs of symbols, whether $T$ accepts 
the convolution of $u$ and $v$ \cite{lohrey-leaf:11}.

MTDDs restrict 2SLPs by forbidding unbalanced derivation trees.
The derivation tree of an MTDD results from unfolding the rules in \eqref{eq-rules-of-MTDD};
it is a tree, where every non-leaf node has exactly four children and every root-leaf path has
the same length. 

Let us finally mention that straight-line programs are also used for the compact
representation of other objects, e.g. polynomials \cite{IbMo83}, trees \cite{LoMa06}, graphs \cite{LeWa92},
and regular languages \cite{GeffertMP10}.

\paragraph*{\bf Tensor circuits.}

In \cite{BeaudryH07,DammHM02}, the authors investigated the problems of evaluating tensor formulas and tensor 
circuits. Let us restrict to the latter. A tensor circuit is a circuit where the gates evaluate to matrices
over a semiring
and the following operations are used: matrix addition, matrix multiplication, and 
tensor product. Recall that the tensor product of two matrices $A = (a_{i,j})_{1 \leq i \leq m, 1 \leq i \leq m}$ and $B$ is the matrix
$$
A \otimes B = \left(\begin{array}{ccc} a_{1,1} B & \cdots &  
      a_{1,m} B \\
    \vdots & & \vdots \\
     a_{n,1} B & \cdots &  
      a_{n,m} B 
      \end{array}\right)  
$$
It is a $(mk \times nl)$-matrix if $B$ is a $(k \times l)$-matrix. 
In \cite{BeaudryH07} it is shown among other results that computing the output value of a scalar tensor circuit
(i.e., a tensor circuit that yields a $(1\times 1)$-matrix) over the natural numbers
is complete for the counting class $\#\mathsf{EXP}$.
An $\MTDD_+$ over $\ZZ$ can be seen as a tensor circuit that (i) does not use matrix multiplication and (ii) where
for every tensor product the left factor is a $(2 \times 2)$-matrix. To see the correspondence, note that
\begin{gather*}
\left(\begin{array}{ll} A_{1,1} &  A_{1,2}\\ A_{2,1} &
      A_{2,2}\end{array}\right)  =
\left(\begin{array}{ll} 1&  0\\  0 &
      0\end{array}\right)    \!\otimes\! A_{1,1} +
\left(\begin{array}{ll} 0 &  1\\  0 &
      0\end{array}\right)    \!\otimes\! A_{1,2} +
\left(\begin{array}{ll} 0 &  0\\  1 &
      0\end{array}\right)    \!\otimes\! A_{2,1} +
      \left(\begin{array}{ll} 0 &  0\\  0 &
      1\end{array}\right)    \!\otimes\! A_{2,2} \\
\left(\begin{array}{ll} a_{1,1} &  a_{1,2}\\ a_{2,1} &
      a_{2,2}\end{array}\right)  \otimes B = 
\left(\begin{array}{ll} a_{1,1} B &  a_{1,2} B \\ a_{2,1} B &
      a_{2,2} B \end{array}\right)      
\end{gather*}
Each of the matrices $a_{i,j}B$ can be generated from $B$ and $-B$ using 
$\log |a_{i,j}|$ many additions (here we use the fact that the underlying semiring is $\ZZ$).

\section{Preliminaries}

We consider matrices over a semiring $(S,+,\cdot)$ with $(S,+)$ a finitely generated commutative monoid with unit $0$.
The unit of the monoid $(S,\cdot)$ is $1$. 
We assume that  $0 \cdot a = a \cdot 0 = 0$ for all $a \in S$. 
Hence, if $|S| > 1$, then $1 \neq 0$ ($0=1$ implies $a = 1 \cdot a = 0 \cdot a = 0$ for all $a \in S$).
With $S^{n \times n}$ we denote the set of all $(n\times n)$-matrices over $S$.

%
All time bounds in this paper implicitly refer to 
the RAM model of computation with a logarithmic cost
measure for arithmetical operations on integers, where
arithmetic operations on $n$-bit numbers need time 
$O(n)$. For a number $n \in \ZZ$ let us denote with $\bin(n)$ its 
binary encoding.

We assume that the reader has some basic 
background in complexity theory, in particular we assume that
the reader is familiar with the classes $\NP$, $\coNP$, and $\PSPACE$.
With $\polyL$ (polylogarithmic space) we denote the class $\bigcup_{k \geq 1} \DSPACE(\log^k(n))$
(which by Savitch's theorem is equal to $\bigcup_{k \geq 1} \NSPACE(\log^k(n))$).

A function $f : \{0,1\}^* \to \{0,1\}^*$ belongs to the class 
$\FSPACE(s(n))$ (resp. $\FTIME(s(n))$) if $f$ can be computed on a deterministic
Turing machine in space (resp., time) $s(n)$.\footnote{The assumption that the input
and output alphabet of $f$ is binary is made here to make
the definitions more readable; the extension to arbitrary finite alphabets is straightforward.}
 As usual, only the space 
on the working tapes is counted. Moreover, the output is written
from left to right on the output tape, i.e., 
in each step the machine either outputs a new symbol
on the output tape, in which case the output head moves
one cell to the right, or the machine does not output
a new symbol in which case the output head does not move.
We define 
\begin{eqnarray*}
\FP &=& \bigcup_{k \geq 1} \FTIME(n^k), \\
\FpolyL &=& \bigcup_{k \geq 1} \FSPACE(\log^k(n)), \\
\FPSPACE &=& \bigcup_{k \geq 1} \FSPACE(n^k).
\end{eqnarray*}
Note that for a function $f \in \FPSPACE$ we have
$|f(w)| \leq 2^{|w|^{O(1)}}$ for every input. 
The function that maps an explicitly given integer matrix (with binary encoded entries) to its determinant
belongs to uniform $\NC^2$ \cite{Coo85} and hence
to $\FSPACE(\log^2(n))$.

We need the following simple lemma, see e.g.
\cite[Lemma~2.1]{LoMa11regular}.

\begin{lemma} \label{PSPACE}
If $f \in \FPSPACE$  and  $L \in \polyL$ then $f^{-1}(L) \in \PSPACE$.
\end{lemma}
The following result can be shown in the same way as Lemma~\ref{PSPACE}:

\begin{lemma} \label{PSPACE-counting}
If $f \in \FPSPACE$ and $g \in \FpolyL$ then the mapping $h$ 
defined by $h(x) = g(f(x))$ for all inputs $x$
belongs to  $\FPSPACE$.
\end{lemma}
The counting class $\mathsf{\#P}$ consists of all functions $f : \{0,1\}^* \to
\mathbb{N}$ for which there exists a nondeterministic polynomial time 
Turing machine $M$ with input alphabet $\Sigma$ such that for all 
$x \in \Sigma^*$, $f(x)$ is the number of accepting computation paths
of $M$ for input $x$. If we replace nondeterministic polynomial time 
Turing machines by  nondeterministic polynomial space 
Turing machines (resp.  nondeterministic logspace 
Turing machines), we obtain the class $\mathsf{\#PSPACE}$ \cite{Ladner89}
(resp. $\mathsf{\#L}$ \cite{AlvarezJ93}).
Note that for a mapping $f \in \mathsf{\#PSPACE}$, the number $f(x)$ 
may grow doubly exponential in $|x|$, whereas for $f \in\mathsf{\#P}$, the number $f(x)$ 
is bounded singly exponential in $|x|$. Ladner \cite{Ladner89} has shown that a mapping
$f : \Sigma^* \to \NN$ belongs to $\mathsf{\#PSPACE}$ if and only if the mapping $x \mapsto \bin(f(x))$ 
belongs to $\FPSPACE$. One cannot expect a corresponding result 
for the class $\mathsf{\#P}$: If for every function $f \in \mathsf{\#P}$ the mapping $x \mapsto \bin(f(x))$ 
belongs to $\FP$, then by Toda's  theorem \cite{To91} the polynomial time hierarchy collapses down to $\Ptime$.
For $f \in \mathsf{\#L}$, the mapping $x \mapsto \bin(f(x))$ belongs to $\mathsf{NC}^2$ and hence to 
$\FP \cap \FSPACE(\log^2(n))$ \cite[Theorem~4.1]{AlvarezJ93}.
The class $\mathsf{GapL}$ (resp., $\mathsf{GapP}$,
$\mathsf{GapPSPACE}$) consists of all differences of two functions in 
$\mathsf{\#L}$ (resp., $\mathsf{\#P}$,
$\mathsf{\#PSPACE}$).
From Ladner's result \cite{Ladner89} it follows easily that a function
$f : \{0,1\}^* \to \ZZ$ belongs to $\mathsf{GapPSPACE}$
if and only if the mapping $x \mapsto \bin(f(x))$ belongs to $\FPSPACE$, see also
\cite[Theorem~6]{GalotaV05}.

Logspace reductions between functions can be defined analogously
to the language case: If $f, g : \{0,1\}^* \to X$ with $X \in \{\mathbb{N},\mathbb{Z}\}$, then $f$ is logspace reducible to $g$ if there exists
a function $h \in \FSPACE(\log n)$ such that
$f(x) = g(h(x))$ for all $x$. 
Toda \cite{Toda91countingproblems} has shown that computing the determinant of a given integer
matrix is  $\mathsf{GapL}$-complete.

\section{Succinct matrix representations} 

In this section, we introduce several succinct matrix representations. We formally
define multi-terminal decision diagrams and their extension by the addition operation.
Moreover, we briefly discuss the representation of matrices by boolean circuits.

\subsection{Multi-terminal decision diagrams} \label{sec-MTDD}

Fix a semiring $(S,+,\cdot)$ with $(S,+)$ a finitely generated commutative monoid, 
and let $\Gamma \subseteq S$ be a finite generating set for $(S,+)$.  Thus, every element
of $S$ can be written as a finite sum $\sum_{a \in \Gamma} n_a a$ with $n_a \in \NN$. 
A {\em multi-terminal decision diagram $G$  with addition ($\MTDD_+$) of
height $h$} is a
triple $(N,P,A_0)$, where $N$ is a finite set of variables which is
partitioned into non-empty sets $N_i$ ($0 \leq i \leq h$),
$N_h = \{A_0\}$ ($A_0$ is called the {\em start variable}), and 
$P$ is a set of rules of the following three forms:
\begin{itemize}
\item  $A\to \left(\begin{array}{ll} A_{1,1} &  A_{1,2}\\ A_{2,1} &
      A_{2,2}\end{array}\right)$ with $A \in N_i$ and  
      $A_{1,1}, A_{1,2}, A_{2,1}, A_{2,2} \in N_{i-1}$ for some $1
      \leq i \leq h$
\item $A \to A_1 + A_2$ with $A,A_1,A_2\in N_i$ for some $0 \leq i
  \leq h$
\item $A \to a$ with $A \in N_0$ and $a \in \Gamma \cup \{0\}$
\end{itemize}
Moreover, for every variable $A \in N$ there is exactly one rule 
with left-hand side $A$, and the relation 
$\{ (A,B) \in N \times N \mid B \text{ occurs in the right-hand side for } A \}$ 
is acyclic. If $A \in N_i$ then we say that $A$ has height $i$.
The $\MTDD_+$ $G$ is called an {\em MTDD} if for every
addition rule  $(A \to A_1+A_2) \in P$ we have $A,A_1,A_2\in N_0$.
In other words, only scalars are allowed to be added. Since we assume
that $(S,+)$ is generated by $\Gamma$, this allows to produce arbitrary elements of 
$S$ as matrix entries.
For every $A \in N_i$ we define a square matrix $\val(A)$ of dimension $2^i$
in the obvious way by unfolding the rules.
Moreover, let $\val(G)=\val(A_0)$ for the start variable $A_0$ of $G$. This is a 
$(2^h \times 2^h)$-matrix.
The size of a rule  $A \to a$ with $a \in \Gamma \cup \{0\}$
is $1$, all other rules have size $\log |N|$. 
The size $|G|$ of the $\MTDD_+$ $G$  is the sum of the sizes of its rules;
this is up to constant factors the length of the binary coding of $G$.
An $\MTDD_+$ $G$ of size $n \log n$ can represent a $(2^n \times 2^n)$-matrix.
Note that only square matrices whose dimension is a power of 2 can be represented. 
Matrices  not fitting this format can be filled up appropriately, depending on the purpose.

An MTDD, where all rules have the form $A \to a \in \Gamma \cup \{0\}$ or $A \to B+C$
generates an element of the semiring $S$. Such an MTDD is an arithmetic
circuit in which only input gates and addition gates are used, and is called a \emph{$+$-circuit} in the following. 
In case the underlying semiring is $\ZZ$, a $+$-circuit with  $n$ variables can produce a number of size $2^n$, and the binary
encoding of this number can be computed in time $\mathcal{O}(n^2)$ from the $+$-circuit (since, we need $n$ additions 
of numbers with at most $n$ bits). In general, for a $+$-circuit over the semiring $S$, we can 
compute in quadratic time numbers $n_a$ ($a \in \Gamma$) such that $\sum_{a \in \Gamma} n_a \cdot a$ 
is the semiring element to which the $+$-circuit evaluates to.

Note that the notion of an $\MTDD_+$ makes sense for commutative monoids, since we only used
the addition of the underlying semiring. But soon, we want to multiply matrices, for which we need a semiring.
Moreover, the notion of an $\MTDD_+$ makes sense in any dimension, here we only defined the 2-dimensional
case.

\begin{xmpl} \label{examples}
It is straightforward to produce the unit matrix $I_{2^n}$ of dimension $2^n$ by an 
MTDD of size $O(n \log n)$:
\begin{equation*}
A_0 \to 1, \  \  0_0 \to 0, \ \
      A_j \to  \left(\begin{array}{ll} A_{j-1} & 0_{j-1}\\ 0_{j-1} &
      A_{j-1}\end{array}\right), \ \
     0_j \to  \left(\begin{array}{ll} 0_{j-1} &  0_{j-1}\\ 0_{j-1} &
     0_{j-1}\end{array}\right) \ \ (1 \leq j \leq n).
\end{equation*}
(the start variable is $A_n$ here). In a similar way, one can produce 
the lower triangular $(2^n \times 2^n)$-matrix, where  entries on the
diagonal and below are $1$.
%
To produce the $(2^n \times 2^n)$-matrix over $\mathbb{Z}$, where all entries
in the $k$-th row are $k$,  we need the following rules:
\begin{align*}
    & E_0 \to 1, \ \ E_j \to \left(\begin{array}{ll} E_{j-1} + E_{j-1} &  E_{j-1} + E_{j-1}\\ E_{j-1}+ E_{j-1}
        & E_{j-1}+ E_{j-1} \end{array}\right)   \ \ (1 \leq j \leq n)\\
    & C_0 \to 1, \  \
    C_j \to \left(\begin{array}{ll} C_{j-1} &  C_{j-1} \\ C_{j-1}+ E_{j-1} 
        & C_{j-1}+ E_{j-1}  \end{array}\right) \ \ (1 \leq j \leq n).
\end{align*}
Here, we are bit more liberal with respect to the format of rules, but the above rules
can be easily brought into the form from the general definition of an $\MTDD_+$.
Note that $E_j$ generates the 
$(2^j \times 2^j)$-matrix with all entries equal to $2^j$, and that $C_n$ generates
the desired matrix.
\end{xmpl}   
Note that the matrix from the last example cannot be produced by an MTDD of polynomial
size, since it contains an exponential number of different matrix entries (for the same reason
it cannot be produced by an 2SLP \cite{berman-karpinski-2d:02}). 
This holds for any non-trivial semiring.

\begin{theorem} \label{prop-MTDD-succinct}
For any semiring with at least two elements,
MTDD$_{+}$ are exponentially more succinct than MTDDs. 
\end{theorem}

\begin{proof} 
For simplicity we argue with MTDDs in dimension 1 (which generate vectors).  We must have $1 \neq 0$ in $S$.
Let $m,d > 0$   be such that  $m = 2^d$. 
For $0 \leq i \leq m-1$ let $A_i$ such that $\val(A_i)$ has length $m$,
the $i$-th entry is $1$ (the first entry is the $0$-th entry) and all other entries are $0$.
Moreover, let $B_i$ such that $\val(B_{i})$ is the concatenation  of $2^i$ copies of $\val(A_i)$.
Let $C_0$ produce the $0$-vector of length $m=2^d$,
and for $0 \leq i \leq m-1$ 
let  $C_{i+1} \to (C_{i},  C_i + B_{i})$.
Then $\val(C_{m})$ is of length $2^{d+m}$ and consists of the concatenation of all binary strings
of length $m$. This $\MTDD_+$ for this vector is of size $O(m^2 \log m)$, whereas an equivalent $\MTDD$  must have size at least $2^m$, since for every binary string
of length $m$ there must exist a nonterminal.
\qed
\end{proof}
The following result shows that the matrix product of two
MTDD-represented matrices may be incompressible with MTDDs.

\begin{theorem} \label{prop-MTDD-mult}
For any semiring 
with at least two elements 
there exist MTDDs $G_n$ and $H_n$  of the same height
$n$ and size $O(n^2 \log n)$
such that $\val(G_n)\cdot \val(H_n)$ can only be represented by an MTDD
of size at least $2^n$. 
\end{theorem}

\begin{proof}
The construction is similar to those in the proof of Theorem~\ref{prop-MTDD-succinct}.
We must have $0 \neq 1$ in $S$.
Let $m = 2^d$.
For $0 \leq i \leq m -1$ let $A_i$ be such that $\val(A_i)$ is the $(m \times m)$-matrix  
with  $\val(A_i)_{1,i+1} = 1$ and all other entries $0$.
Define $B_{i,0}$ by $B_{i,0} \to A_i$ and 
$$B_{i,j} \to \left(\begin{array}{ll} B_{i,j-1}& B_{i,j-1} \\ 0 & 0 \end{array}\right)$$
for $1 \leq j \leq i$. Then $\val(B_{i,i})$ is the $(2^{d+i} \times 2^{d+i})$-matrix, where
the first row is the vector $\val(B_i)$ from the proof of Theorem~\ref{prop-MTDD-succinct},
and all other entries are $0$. Finally add nonterminals $C_0, \ldots, C_m$, where
$\val(C_0)$ is the $(m \times m)$-matrix with all entries $0$ and 
$$C_{i+1}  \to \left(\begin{array}{ll} C_{i} & C_{i}\\0& B_{i,i}\end{array}\right)$$
$0 \leq i \leq m-1$. In this way we obtain an MTDD for the $(2^{m+d} \times 2^{m+d})$-matrix 
$\val(C_m)$ of size $O(m^2 \log m)$.
This matrix contains $1$ in the $i$-th column if and only if the $i$-th entry in 
the vector $\val(C_m)$ from the proof of Theorem~\ref{prop-MTDD-succinct}  is $1$.  
Moreover, no column of $\val(C_m)$ contains more than one $1$-entry.
Hence, the product of the  $(2^{m+d} \times 2^{m+d})$-matrix where every entry is $1$ with $\val(C_m)$ 
 a matrix where every row is the vector $\val(C_m)$ from the proof of Theorem~\ref{prop-MTDD-succinct}.
 \qed
\end{proof}
On the other hand, the product of two $\MTDD_+$-represented
matrices can be represented 
by a polynomially sized  $\MTDD_+$:

\begin{theorem}\label{thm:main-matrix-mult}
For $\MTDD_+$ $G_1$ and $G_2$ of the same height one can compute in time $O(|G_1|
\cdot |G_2|)$ an $\MTDD_+$ $G$ of size $O(|G_1|\cdot |G_2|)$ with
$\val(G) = \val(G_1) \cdot \val(G_2)$.
\end{theorem}

\begin{proof}
Recall that $\Gamma$ is a finite generating set for the additive monoid of our underlying semiring $S$.
For all pairs $(a,b) \in \Gamma \times \Gamma$, we can write down a  $+$-circuit of constant size that computes
$ab$, let $S_{a,b}$ its start variable.

Given two  $\MTDD_+$ $G_1$ and $G_2$, we compute a
new $\MTDD_+$ $G$ that contains for all variables $A$ of $G_1$ and
$B$ of $G_2$ of the same height a variable $(A,B)$ such that $\val_G(A,B) = \val_{G_1}(A)
\cdot \val_{G_2}(B)$. So, let $A$ and $B$ be variables of $G_1$
and $G_2$, respectively, of the same height.
\begin{enumerate}
  \item If $A$ and $B$ are of height 0 and the corresponding rules are
    $A \to a$, $B \to b$ with $a,b \in \Gamma \cup \{0\}$, then the rule for $(A,B)$
    is $(A,B) \to S_{a,b}$ (actually, we should replace $S_{a,b}$ by its corresponding right-hand side).
  \item If the rule for $A$ is of the form 
        $A \to A_1 + A_2$, then we add the rule $(A,B) \to (A_1,B) +
        (A_2,B)$ to $G$.
  \item If the right-hand side for $A$ is not a sum but the rule for
    $B$ is of the form $B \to B_1 + B_2$, then we add the rule $(A,B) \to (A,B_1) +
        (A,B_2)$ to $G$.
  \item Finally, assume that neither the right-hand side for $A$ nor
    for $B$ is a sum or an explicit integer. Then the rules for $A$ and
    $B$ have the form
    $$A \to \left(\begin{array}{ll} A_{1,1} & A_{1,2}\\ A_{2,1} &
        A_{2,2}\end{array}\right) \text{ and } 
      B \to \left(\begin{array}{ll} B_{1,1} & B_{1,2}\\ B_{2,1} &
          B_{2,2}\end{array}\right).
    $$
    Then we add the following rules to $G$:
    \begin{gather*}
    C_{i,j} \to (A_{i,1},B_{1,j}) + (A_{i,2},B_{2,j})  \text{ for } 1
    \leq i,j \leq 2 \\
    (A,B) \to \left(\begin{array}{ll} C_{1,1} & C_{1,2}\\ C_{2,1} & C_{2,2}\end{array}\right)
    \end{gather*}
\end{enumerate}
Clearly, if $S_i$ is the start variable of $G_i$, then 
$\val_G(S_1,S_2) = \val(G_1) \cdot \val(G_2)$.
The bound from the theorem for the construction and size of $G$
follows immediately from the construction. Note that every rule $C \to
c$ of $G_i$ with $c \in \ZZ$
contributes $\log|c|$ to the size of $G_i$. Hence in time 
$O(|G_1|\cdot |G_2|)$ we can compute all products $ab$ for rules
$A \to a$ and $B \to b$ of $G_1$ and $G_2$, respectively.
\qed
\end{proof}
The following proposition presents several further matrix operations that 
can be easily implemented in polynomial time for an $\MTDD_+$-represented input matrix.

\begin{proposition}\label{prop:simple-constructions} 
Let $G,H$  be a  $\MTDD_+$ with $|G|=n$, $|H|=m$,  and $1 \leq i,j \leq 2^{\mathsf{height}(G)}$
\begin{enumerate}[(1)]
 \item An $\MTDD_+$ for the transposition of $\val(G)$ can be
    computed in time $O(n)$.
  \item $+$-circuits for 
   the sum of all entries of $\val(G)$ and the trace of $\val(G)$ can be computed in time  $O(n)$.
   \item A $+$-circuit for the matrix entry  $\val(G)_{i,j}$ can be computed in time $O(n)$. 
   \item $\MTDD_+$ of size $O(n \cdot m)$ for the tensor product $\val(G) \otimes \val(H)$ (which includes the scalar product)
   and the element-wise (Hadamard) product
   $\val(G) \circ \val(H)$ (assuming $\height(G) = \height(H)$)
   can be computed in  time $O(n \cdot m)$.
\end{enumerate}
\end{proposition}

\begin{proof}
Point (1) (transposition): 
We  replace every rule in $G$ of the form 
  \begin{equation} \label{computing-trace}
  A \to \left(\begin{array}{ll} A_{1,1} &  A_{1,2} \\   A_{2,1} &
        A_{2,2} \end{array}\right)
  \end{equation}
  by the rule 
  \begin{equation*} 
  A \to \left(\begin{array}{ll} A_{1,1} &  A_{2,1} \\   A_{1,2} &
        A_{2,2} \end{array}\right).
  \end{equation*} 
  Point (2):   
  The sum of all entries of $\val(G)$ can be represented by the $+$-circuit
  that contains all rules $A \to A_{1,1}+A_{1,2}+A_{2,1}+A_{2,2}$ for $G$-rules
  of the form \eqref{computing-trace}.
 Similarly, we can compute a $+$-circuit for the trace of $\val(G)$ by
 replacing every rule \eqref{computing-trace} by $A \to A_{1,1} +
 A_{2,2}$.

\medskip
\noindent
Point (3): 
We transform the $\MTDD_+$ $G$ into a $+$-circuit $G'$ with the same set of variables
such that $\val(G') = (\val(G))_{i,j}$. Let $(i_h \cdots i_1)$ and $(j_h \cdots j_1)$ the binary
expansions if $i-1$ and $j-1$ (numbers in the range $[0,2^{\mathsf{height}(G)}-1]$), respectively, where $i_h$ and $j_h$ are the most significant bits.
Here, we add leading zeros on the left so that both numbers have exactly $h$ bits.

Now we can define the rules of the $+$-circuit $G'$.
Rules of the form $A \to a$ with $a \in \mathbb{Z}$ 
and $A \to A_1+A_2$ are simply copied to $G'$.
For a rule of the form 
$$
A \to \left(\begin{array}{ll} A_{0,0} &  A_{0,1}\\ A_{1,0} & A_{1,1}\end{array}\right).
$$ 
where $A$ has height $k$ we add to $G'$ the rule 
$A \to A_{i_k,j_k}$.

\medskip
\noindent
Point (4): For every variable $C$ of $G$ and every variable $D$ of $H$
 let  $(C,D)$ be a new variable of height $\mathsf{height}(C) + \mathsf{height}(D)$. We define the rule 
for $(C,D)$ in such a way that $\val(C,D) = \val(C) \otimes \val(D)$.
The rules reflect the bilinearity of the tensor product.

If $C \to a$ and $D \to b$  for $a,b \in \Gamma$, then $(C,D) \to S_{a,b}$, where
$S_{a,b}$ is the start variable for a (constant size) $+$-circuit that computes $a \cdot b$.

Now assume that $C \to a$ but the rule for $D$ is not terminal. If
$D \to D_1+D_2$, then $(C,D) \to (C,D_1) + (C,D_2)$ and 
if 
$$
D  \to \left(\begin{array}{ll} D_{1,1} &  D_{1,2} \\   D_{2,1} &
        D_{2,2} \end{array}\right)
$$
then
$$
(C,D)  \to \left(\begin{array}{ll} (C,D_{1,1}) &  (C,D_{1,2}) \\   (C,D_{2,1}) &
        (C,D_{2,2}) \end{array}\right).
$$
Finally, assume that the rule for $C$ is not terminal. If 
$C \to C_1+C_2$, then $(C,D) \to  (C_1,D) +  (C_2,D)$, and if 
$$C  \to \left(\begin{array}{ll} C_{1,1} &  C_{1,2} \\   C_{2,1} &
        C_{2,2} \end{array}\right),
$$ 
then $$(C,D) \to  \left(\begin{array}{ll} (C_{1,1},D) &  (C_{1,2},D) \\   (C_{2,1},D) &
        (C_{2,2},D) \end{array}\right).
$$       
 The proof for the construction of the element-wise product is similar as for the tensor-product.
\qed
\end{proof}

\subsection{Boolean circuits} \label{sec-bool-circuits}

Another well-studied succinct representation are boolean circuits \cite{GaWi83}.
A boolean circuit with $n$ inputs represents a binary string
of length $2^n$, namely the string of output values for 
the $2^n$ many input assignments (concatenated in lexicographic
order). In a similar way, we can use circuits to encode large matrices. We propose
two alternatives:    

A boolean circuit $C(\overline{x}, \overline{y}, \overline{z})$ with
$|\overline{x}|=m$ and $|\overline{y}|=|\overline{z}| = n$ encodes 
a $(2^n \times 2^n)$-matrix $M_{C,2}$ with integer entries bounded by
$2^{2^m}$ that is defined as follows: For all $\overline{a} \in \{0,1\}^m$
and $\overline{b}, \overline{c} \in \{0,1\}^n$, the $\overline{a}$-th bit
(in lexicographic order)
of the matrix entry at position $(\overline{b}, \overline{c})$ in $M_C$ is $1$ 
if and only if $C(\overline{a}, \overline{b}, \overline{c})=1$.

Note that in contrast to $\MTDD_+$, the size of an entry in $M_{C,2}$ can 
be doubly exponential in the size of the representation $C$ (this is the reason for the index $2$
in $M_{C,2}$). The following alternative  is closer to $\MTDD_+$:
A boolean circuit $C(\overline{x}, \overline{y})$ with
$|\overline{x}|=|\overline{y}| = n$ and $m$ output gates
encodes  a $(2^n \times 2^n)$-matrix $M_{C,1}$ with integer entries bounded by
$2^m$ that is defined as follows: For all $\overline{a}, \overline{b} \in \{0,1\}^n$,
$C(\overline{a}, \overline{b})$ is the binary encoding of the entry at position $(\overline{a}, \overline{b})$ in $M_C$.

Circuit representations for matrices are at least as 
succinct as $\MTDD_+$.  More precisely,
from a given $\MTDD_+$ $G$ one can compute in logspace
a Boolean circuit $C$ such that $M_{C,1} = \val(G)$. This is a direct corollary of 
Proposition~\ref{prop:simple-constructions}(3) (stating that a given entry of 
an $\MTDD_+$-represented matrix can be computed in polynomial time) and
the fact that polynomial time computations can be simulated by boolean circuits.
Recently, it was shown that checking whether for a given circuit $C$ 
the determinant of the matrix $M_{C,1}$ 
vanishes is $\PSPACE$-complete \cite{GreKoiPort12}. An algebraic
version of this result for the algebraic complexity class $\mathsf{VPSPACE}$ is shown in \cite{Mal11}. 
Theorem~\ref{theorem-det} from Section~\ref{sec-hard} will strengthen the result from \cite{GreKoiPort12} to 
MTDD-represented matrices.

\section{Testing equality} \label{sec-eq}

In this section, we consider the problem of testing equality of $\MTDD_+$-represented matrices.
For this, we do not need the full semiring structure, but we only need the finitely generated additive  
monoid $(S,+)$. We will show that equality can be checked in polynomial time if $(S,+)$ is cancellative
and {\sf coNP}-complete otherwise.

First we consider the case of a  finitely generated abelian group.
The  proof of the following lemma involves only basic linear algebra.

\begin{lemma}\label{lemma:equations-torsionfree-redundant}
Let $a_{i,1}x_1 + \cdots +a_{i,n}x_n = 0$ for $1 \leq i \leq m \leq n+1$ be equations over 
a torsion-free abelian group $A$, where $a_{i,1},\ldots, a_{i,n} \in
\ZZ$, and the variables $x_1,\ldots,x_n$ range over $A$.
One can determine in time polynomial in $n$ and 
$\max\{ \log |a_{i,j}| \mid  1 \leq i\leq m, 1 \leq j \leq n\}$
an equivalent set of at most $n$ linear equations.   
\end{lemma}

\begin{proof}
Let $a_i = (a_{i,1}, \ldots, a_{i,n}) \in \mathbb{Z}^n$
be the vector of coefficients of the $i$-th equation.
For $0 \leq i \leq n$ let $U_i \subseteq \mathbb{Q}^n$
be the subspace of the vector space generated by $a_1,\ldots,a_i$
($U_0$ is the 0-space). 
For $i=1,\ldots,n+1$, we now test whether $a_i \in U_{i-1}$.
This can be checked by testing whether a system of linear equations
has a solution in $\mathbb{Q}^n$. This problem can be solved in 
time polynomial in $n$ and $\log(\max\{ |a_{i,j}| \mid  1 \leq i\leq
m, 1 \leq j \leq n\})$, e.g. by Gaussian elimination.
If $a_i \in U_{i-1}$ then we obtain an equation
$$
\lambda_i a_i =  \lambda_1 a_1 + \cdots + \lambda_{i-1} a_{i-1} 
$$
with $\lambda_1, \ldots, \lambda_i \in \mathbb{Z}$ and $\lambda_i \neq 0$.
Hence, if group elements $x_1, \ldots, x_n \in A$ satisfy
$a_{j,1}x_1 + \cdots +a_{j,n}x_n = 0$ for all $1 \leq j \leq i-1$, then
we get 
$\lambda_i (a_{i,1}x_1 + \cdots +a_{i,n}x_n) = 0$
in $A$. Since $A$ is assumed to be torsion-free, we get 
$a_{i,1}x_1 + \cdots +a_{i,n}x_n = 0$. Hence, 
the $i$-th equation is redundant.
Moreover, there must be an $1 \leq i \leq n+1$ with $a_i \in U_{i-1}$:
If $a_i \not\in U_{i-1}$ for $1 \leq i \leq n$, then $a_1,\ldots,a_n$
are linearly independent and therefore generate the full
$\mathbb{Q}^n$. But then $a_{n+1} \in U_n$.
\qed
\end{proof}
Recall that the {\em exponent} of an abelian group $A$ is the smallest integer $k$ (if it exists) such that
$kg = 0$ for all $g \in A$. The following result is shown in \cite{storjohann-mulders:98}:

\begin{lemma}\label{lemma:equations-torsion-redundant}
Let $k \geq 2$ and let $A$ be an abelian group of exponent $k$.
Let $a_{i,1}x_1 + \cdots +a_{i,n}x_n = 0$ for $1 \leq i \leq m \leq n+1$ be equations, where $a_{i,1},\ldots, a_{i,n} \in
\ZZ$, and the variables $x_1,\ldots,x_n$ range over $A$.
Then one can determine in time polynomial in $n$, $\log(k)$, and 
$\max\{ \log |a_{i,j}| \mid  1 \leq i\leq m, 1 \leq j \leq n\}$
an equivalent set of at most $n$ linear equations.  
\end{lemma}

\begin{proof} 
We can consider the coefficients $a_{i,j}$ as elements from $\ZZ_k$.
By  \cite{storjohann-mulders:98} we can compute the Howell normal form 
of the matrix $(a_{i,j})_{1 \leq i \leq n+1, 1 \leq j \leq n} \in \ZZ_k^{(n+1) \times n}$ in polynomial time.
The Howell normal form is an $(n \times n)$-matrix with the same row span (a subset of the module 
$\ZZ_k^n$) as the original matrix, and hence defines an equivalent set of 
linear equations.
\qed
\end{proof}
 
\begin{theorem}\label{thm:matrix-equality} Let $G$ be an $\MTDD_+$ over a
finitely generated  abelian group $S$. 
Given two different variables $A_1,A_2$ of the same height, 
it is possible to check $\val(A_1) = \val(A_2)$ in time polynomial in $|G|$.
\end{theorem}

\begin{proof} Since every finitely generated group is a finite direct product of copies
of $\ZZ$ and $\ZZ_k$ ($k \geq 2$), it suffices to prove the theorem only for these
groups. 

Consider the case $S = \ZZ$. 
The algorithm stores a system of $m$ equations ($m$ will be bounded later) of the form 
$a_{i,1}B_1 + \cdots + a_{i,k}B_k = 0$, where all $B_1,\ldots,B_k$ are 
pairwise different variables of the same height $h$. We treat the 
variables $B_1,\ldots,B_k$ as variables that range over the torsion-free abelian group
$\mathbb{Z}^{2^h \times 2^h}$.
We start with the single equation $A_1 - A_2 = 0$.
We use the rules of $G$ to transform the system of equations into 
another system of equations whose variables have strictly smaller height.
Assume the current height is $h > 1$. We
iterate the following steps until only variables of height $h-1$ occur in
the equations:

\smallskip
\noindent
{\it Step 1.}
Standardize equations: Transform all equations into the form 
   $a_1B_1 + \cdots + a_mB_m = 0$, where the $B_i$ are different variables
   and the $a_i$ are integers. 

\smallskip
\noindent
{\it Step 2.}
Reduce the number of equations,
      using Lemma~\ref{lemma:equations-torsionfree-redundant} applied to the torsion-free abelian
      group $\mathbb{Z}^{2^h \times 2^h}$.

\smallskip
\noindent
{\it Step 3.} 
If a variable $A$ of height $h$ occurs in the
    equations, and the rule for $A$ has  
    the form $A \to A_1 + A_2$, then replace every occurrence of $A$ in the equations 
     by $A_1+A_2$.

\smallskip
\noindent
{\it Step 4.}
If none of steps 1--3 applies to the equations, then
     only  rules of the form
    \begin{equation}\label{decompose-equations}  
    A \to \left(\begin{array}{ll} A_{1,1} &  A_{1,2} \\   A_{2,1} &
        A_{2,2} \end{array}\right)
     \end{equation} 
     are applicable 
     to a variable $A$ (of height $h$) occurring in the equations.
     Applying all possible  rules of this form for the current
     height results in a set of equations where all variables 
     are $(2\times 2)$-matrices over variables of height $h-1$
     (like the right-hand side of \eqref{decompose-equations}). Hence,
     every equation can be decomposed into 4 equations, where all
     variables are variables of height $h-1$.

 If the height of all variables is finally 0, then only rules of the form $A \to a$ are applicable. 
 In this case, replace all variables by the corresponding  integers,
 and check whether all resulting equations are valid or not.  If all
 equations hold, then the input equation holds, i.e., $\val(A_1)=\val(A_2)$. Otherwise, 
 if  at least one equation is not valid, then $\val(A_1) \neq \val(A_2)$.
  
The number of variables in the equations is bounded by the number of variables of $G$.
An upper bound on the  absolute value of the coefficients in the equations is $2^{|G|}$, 
since only iterated addition can be performed to increase the coefficients.
Lemma~\ref{lemma:equations-torsionfree-redundant} shows that the number of equations after step~2 above is at most $|G|$, 
(the bound for the number of different variables). 

For the case $S = \ZZ_k$ the same procedure works, we only have to use Lemma~\ref{lemma:equations-torsion-redundant} instead of 
Lemma~\ref{lemma:equations-torsionfree-redundant}.
 \qed
\end{proof}

\begin{corollary}\label{coro:matrix-equality} Let $M$ be a finitely generated
cancellative commutative monoid.
Given an $\MTDD_+$ $G$ over $M$ and two variables $A_1$ and $A_2$ of $G$, one can check  $\val(A_1) = \val(A_2)$ in time polynomial in $|G|$.
\end{corollary}

\begin{proof}
A cancellative commutative monoid $M$ embeds into its Grothendieck group $A$, which is the 
quotient of $M \times M$ by the congruence defined by $(a,b) \equiv (c,d)$ if and only if 
$a+d = c+b$ in $M$. This is an abelian group, which is moreover finitely generated if $M$ is finitely
generated. Hence, the result follows from Theorem~\ref{coro:matrix-equality}.
\qed
\end{proof}
Let us now consider non-cancellative commutative monoids:

 \begin{theorem} \label{thm-identity-coNP}
 Let $M$ be a non-cancellative finitely generated commutative monoid. It is {\sf coNP}-complete
 to check $\val(A_1) = \val(A_2)$ for a given $\MTDD_+$ $G$ over $M$ and two variables $A_1$ and $A_2$ of $G$.
 \end{theorem}
 
 \begin{proof}
 We start with the upper bound. Let $\{a_1, \ldots, a_k\}$ be a finite generating set of $M$.
 Let $G$ be an $\MTDD_+$ over $M$ and let $A_1$ and $A_2$ two variables of $G$.
 Assume that $A_1$ and $A_2$ have the same height $h$. It suffices to check in polynomial time for two given indices $1 \leq i,j \leq 2^h$
 whether $\val(A_1)_{i,j} \not= \val(A_2)_{i,j}$. 
 From $1 \leq i,j \leq 2^h$ we can compute $+$-circuits for the matrix entries $\val(A_1)_{i,j}$ and $\val(A_2)_{i,j}$. From these
 circuits we can compute numbers $n_1, \ldots, n_k, m_1, \ldots, m_k \in \mathbb{N}$ in binary representation such that
 $\val(A_1)_{i,j} = n_1 a_1 + \cdots + n_k a_k$ and 
 $\val(A_2)_{i,j} = m_1 a_1 + \cdots + m_k a_k$.
 Now we can use the following result from \cite{Tai68}: There is a semilinear subset $S \subseteq \mathbb{N}^{2k}$ 
 (depending only on our fixed monoid $M$)
 such that for all $x_1, \ldots, x_k, y_1, \ldots, y_k \in \mathbb{N}$ we have: 
 $x_1 a_1 + \cdots + x_k a_k = y_1 a_1 + \cdots + y_k a_k$ if and only if $(x_1,\ldots, x_k, y_1, \ldots, y_k) \in S$.
 Hence, we have to check, whether $v =: (n_1,\ldots, n_k, m_1, \ldots, m_k) \in S$. The semilinear set $S$ is a finite union
 of linear sets. Hence, we can assume that $S$ is linear itself. Let
 $$
 S = \{ v_0 + \lambda_1 v_1 + \cdots + \lambda_l v_l \mid \lambda_1, \ldots, \lambda_l \in \mathbb{N} \},
 $$
 where $v_0, \ldots, v_l \in \mathbb{N}^{2k}$. Hence, we have to check, whether there exist $\lambda_1, \ldots, \lambda_l \in \mathbb{N}$
 such that
 $v = v_0 + \lambda_1 v_1 + \cdots \lambda_l v_l$.
 This is an instance of integer programming in the fixed dimension $2k$, which can be solved in polynomial time \cite{Len83}.
 
 For the lower bound we take elements $x,y,z \in M$ such that $x \neq y$ but $x+z = y+z$. These elements exist since
 $M$ is not cancellative.  We use an encoding of 3SAT from
\cite{berman-karpinski-2d:02}. Take a 3CNF formula
$C=\bigwedge_{i=1}^m C_i$ over $n$ propositional variables $x_1,\ldots, x_n$, and let
$C_i = (\alpha_{j_1} \vee \alpha_{j_2} \vee \alpha_{j_3})$,
where $1 \leq j_1 < j_2 < j_3 \leq n$ and every $\alpha_{j_k}$ is either $x_{j_k}$ or 
$\neg x_{j_k}$.  For every $1 \leq i \leq m$ we define an MTDD $G_i$ as follows:
The variables are $A_0,\ldots, A_n$, and $B_0, \ldots, B_{n-1}$,
where $B_i$ produces the vector of length $2^i$ with all entries equal to $0$
(which corresponds to the truth value {\sf true}, whereas $z \in M$ corresponds to the truth
value {\sf false}).
For the variables $A_0,\ldots, A_n$ we add the following rules:
For every $1 \leq j \leq n$ with $j \not\in \{j_1, j_2, j_3\}$ we take the rule $A_j \to (A_{j-1}, A_{j-1})$.
For every  $j \in \{j_1, j_2, j_3\}$ such that $\alpha_j = x_j$  (resp. $\alpha_j =  \neg x_j$)
we take the rule
$$
A_j \to (A_{j-1}, B_{j-1})  \ \ \text{( resp. } A_j \to (B_{j-1},A_{j-1}) ).
$$
Finally add the rule $A_0 \to z$ and let
$A_n$ be the start variable of $G_i$.
Moreover, let $G$ (resp. $H$) be the 1-dimensional MTDD that produces the vector consisting of $2^n$ many $x$-entries (resp. $y$-entries).
Then, $\val(G) + \val(G_1) + \cdots + \val(G_m) = \val(H) + \val(G_1) + \cdots + \val(G_m)$ if and only if $C$ is unsatisfiable.
\qed
 \end{proof}
 It is worth noting that in the above proof for {\sf coNP}-hardness, we use addition only at the top level in a non-nested way.

%
%
%

\section{Computing determinants and matrix powers}\label{sec-hard}

In this section we present several completeness results for MTDDs over the rings $\ZZ$
and $\ZZ_n$ ($n \geq 2$). It turns out that over these rings,
computing determinants, iterated matrix products, or matrix powers are infeasible
for MTDD-represented input matrices, assuming standard assumptions from complexity
theory. All completeness results in this section are formulated for MTDDs, but they remain
valid if we add addition.
In fact, all upper complexity bounds in this section even hold for 
matrices that are represented by circuits as defined in Section~\ref{sec-bool-circuits}.

All hardness results in this section rely on the fact that
the adjacency matrix of the configuration graph of a 
polynomial space bounded machine can be produced by a small
MTDD (with terminal entries $0$ and $1$), see Section~ \ref{sec-config-graph}.
This was also shown in  \cite[proof of Theorem~7]{feigenbaum-kannan-vardi-viswanathan:99} in the context of OBDDs.
We will prove this fact using an automata theoretic framework that we introduce in Section~\ref{sec-auto-framework}. 
This framework will simplify
the technical details in the proofs in Sections~\ref{sec-det-hard} and \ref{sec-powers-hard}.

\subsection{Layered automata and MTDDs} \label{sec-auto-framework}

In the following we will use some standard notations concerning 
finite automata. 
A {\em layered DFA (deterministic finite automaton) of depth} $m$
is an acyclic DFA $A$ for which the state set $Q$ 
is partitioned into $m+1$ layers $Q_0, \ldots, Q_m$ such that:
\begin{itemize}
\item  $Q_0$ only contains the initial state $q_0$ of $A$. 
\item $Q_m$ only contains two states, one of which is the unique final
state of $A$.
\item Every transition goes from layer $Q_i$ to $Q_{i+1}$ for some $0
\leq i < m$. 
\item For every state $q \in Q_i$ ($1 \leq i < m$) and every input letter
$a$ there exists an $a$-labeled transition from $q$ to 
a state from layer $Q_{i+1}$.  
\end{itemize}
The {\em convolution} of a string $u = a_1\cdots a_n \in \Sigma^*$ 
and a string  $v = b_1\cdots b_n  \in \Gamma^*$ is the string 
$u \otimes v = (a_1,b_1)\cdots (a_n,b_n)$  over the alphabet
$\Sigma \times \Gamma$.  
A layered DFA $A$ of depth $m$ with input alphabet
$\{0,1\}\times\{0,1\}$
defines the directed graph $\mathcal{G}(A)$ with node set $\{0,1\}^m$ (all binary 
strings of length $m$) and an edge from $u \in\{0,1\}^m$ to $v \in\{0,1\}^m$
if and only if $u \otimes v \in L(A)$. So, $A$ recognizes the edge relation of $\mathcal{G}(A)$.
Layered DFAs over the paired alphabet $\{0,1\}\times\{0,1\}$
are basically the same as MTDDs over $\{0,1\}$ (or OBDDs with the interleaved variable ordering):

\begin{lemma} \label{lemma-layered}
One can construct in logspace 
from a given layered DFA $A$ over the paired alphabet 
$\{0,1\}\times\{0,1\}$ an MTDDs $G$ over $\{0,1\}$ such that $\val(G)$ 
is the adjacency matrix of the graph $\mathcal{G}(A)$, and vice versa.
\end{lemma}

\begin{proof}
The variables of $G$ are the states of the automaton $A$, and the start
variable is the initial state $q_0$. Let $P_0, \ldots, P_k$
be the layers of $A$ and let $P_k = \{ p_0,p_1\}$, where
$p_1$ is the final state of $A$.
First, we add the transitions $p_i \to i$ for $i \in \{0,1\}$ to $G$.
Next, let $p \in P_i$ for some $i<k$ and let
$p \xrightarrow{(a,b)} p_{a,b}$ for $a, b \in \{0,1\}$
be the four outgoing transitions from state $p$.
Then we add the rule
$$
p \to \left( \begin{array}{cc}
p_{0,0} \ & \ p_{0,1}  \\
p_{1,0} & p_{1,1} \end{array} \right)
$$
to $G$. The reverse transformation works similarly.
\qed
\end{proof}

\subsection{Generating the configuration graph of a Turing machine by an MTDD}   \label{sec-config-graph}

Let $M$ be a nondeterministic Turing machine (NTM).
Let $Q$ be the set of states of $M$, and let $\Gamma$ be the tape
alphabet of $M$, where $Q \cap \Gamma = \emptyset$.
As usual, configurations of $M$ are encoded as words from
$\Gamma^*Q\Gamma^*$. For two configurations $c_1,c_2 \in
\Gamma^*Q\Gamma^*$ we write $c_1 \vdash_M c_2$ if $M$ can move in one
transition from configuration $c_1$ to configuration $c_2$.
Let us fix an injective encoding $f_M : Q \cup \Gamma \to \{0,1\}^{k_M}
\setminus 0^*$, which is extended to a homomorphism 
from $(Q \cup \Gamma)^*$ to $\{0,1\}^*$. Here, $k_M$ is a large enough constant.
 We exclude words only consisting of $0$'s from the range of $f_M$
 for technical reasons.
The following proposition makes use of the folklore fact (see e.g. the work
on automatic structures) that 
a Turing machine transition only locally modifies
the current configuration and that this local modification can be recognized
by a finite automaton. This locality is not destroyed by an application of the coding
function $f_M$:

\begin{lemma} \label{reach}
Let $M$ be a fixed NTM. For $m \in \mathbb{N}$, one can compute in space $O(\log m)$ 
a layered DFA $A(M,m)$ of depth $k_M(m+1)$
over the paired alphabet
$\{0,1\}\times\{0,1\}$ such that
$L(A(M,m)) = \{ f_M(c_1) \otimes f_M(c_2) \mid c_1, c_2  \in \Gamma^*Q\Gamma^*,
|c_1| = |c_2|=m+1, c_1 \vdash_M c_2 \}$.
\end{lemma}

\begin{proof}
Due to the local nature of Turing machines, there exists a fixed
DFA $A(M)$ over the alphabet
$\{ (0,0), (0,1), (1,0), (1,1)\}$ such that
$$
L(A(M)) = \{ f_M(c_1) \otimes f_M(c_2) \mid c_1, c_2  \in \Gamma^*Q\Gamma^*,
|c_1| = |c_2|, c_1 \vdash_M c_2 \} .
$$
Using the classical product construction, we intersect this automaton
with a layered DFA of depth $k_M(m+1)$ 
for the language $\{0,1\}^{k_M(m+1)}
\otimes \{0,1\}^{k_M(m+1)}$. Such an automaton can be constructed in
logspace. By adding dummy states to the resulting product automaton, we obtain
a layered DFA with the desired properties.
\qed
\end{proof}
For the layered DFA $A(M,m)$ from
Lemma~\ref{reach}, the graph $\mathcal{G}(A(M,m))$
is the configuration graph of $M$ on configurations of tape length $m$.
With Lemma~\ref{lemma-layered} we can compute in space $\log m$ 
an MTDD for the adjacency matrix of
this configuration graph.

\subsection{Hardness of the determinant for MTDDs}  \label{sec-det-hard}

Recall that the determinant of a matrix 
$A = (a_{i,j})_{1 \leq i,j \leq n}$ (over any ring) can be computed as follows, where $S_n$ denotes the set of all
permutations on $\{1,\ldots,n\}$:
$$
\det(A) = \sum_{\sigma\in S_n} \text{sgn}(\sigma) \cdot \prod_{i=1}^n A_{i,\sigma(i)}.
$$
Here, $\text{sgn}(\sigma)$ denotes the signum of the permutation
$\sigma$, which is $1$ (resp., $-1$) if $\sigma$ is a product of an
even (resp., odd) number of transpositions.
If $A$ is the adjacency matrix of a directed graph
$\mathcal{G}$, then we can compute $\det(A)$ by taking the sum
over all cycle covers of $\mathcal{G}$ (a cycle cover of $\mathcal{G}$
is a subset of the edges of $\mathcal{G}$ such that the corresponding
subgraph is a disjoint union of directed cycles), where each cycle
cover contributes to the sum by the signum of the corresponding permutation.
Recall that $\det(A) \neq 0$ if and only if $A$ is invertible.
The value $\det(\val(G))$ for an MTDD $G$ may be of doubly exponential size (and hence
needs exponentially many bits):  The diagonal $(2^n\times 2^n)$-matrix with $2$'s on the diagonal  
has  determinant  $2^{2^n}$.

By the next theorem, computing the determinant of an MTDD-represented
matrix is indeed difficult. To prove this result we use a reduction 
of Toda showing that computing the determinant of an explicitly given 
integer matrix is $\mathsf{GapL}$-complete \cite{Toda91countingproblems}
(which in turn is based on Valiant's classical construction for the universality of the
determinant \cite{Valiant79a}).
We apply this reduction to configuration graphs of polynomial space bounded Turing machines,
whose adjacency matrices can be produced by small MTDDs.

\begin{theorem} \label{theorem-det}
The following holds for every ring $S \in \{ \ZZ\} \cup \{ \ZZ_n \mid n \geq 2\}$:
\begin{enumerate}[(1)]
\item The set $\{ G \mid G \text{ is an MTDD over } S, \det(\val(G)) = 0 \}$ is
$\PSPACE$-complete.
\item
The function $G \mapsto \det(\val(G))$ with $G$ an MTDD over $\ZZ$
is  $\mathsf{GapPSPACE}$-complete.
\end{enumerate}
\end{theorem}

\begin{proof}
Let us start with the upper bounds.
Membership in $\PSPACE$ in statement (1)
can be shown as follows:
Since the determinant of an explicitly given integer matrix can 
be computed in $\FSPACE(\log^2(n))$, one
 can check in $\DSPACE(\log^2(n))$ whether the determinant
of an explicitly given integer matrix is zero.
Moreover, from a given MTDD $G$ we can compute 
the matrix $\val(G)$ in polynomial space. For this, it suffices to compute for 
$G$ and given positions $i,j$ the entry $\val(G)_{i,j}$ in $\PSPACE$; then we can iterate
over all matrix positions $(i,j)$. Actually, a specific matrix entry $\val(G)_{i,j}$ can be even
computed in polynomial time by Theorem~\ref{prop:simple-constructions}(3). 
Membership in $\PSPACE$ for MTDD follows
from Lemma~\ref{PSPACE}. Note that the same argument even applies for matrices
that are represented by boolean circuits in the sense of Section~\ref{sec-bool-circuits}.

The upper bounds in (2) can be shown in the same way using 
Lemma~\ref{PSPACE-counting} and the fact that computing the determinant
of an explicitly given integer matrix with binary
coded integer entries is in $\mathsf{GapL}$.

Let us now prove the lower bound. We start with (1). Let us take a deterministic polynomial space bounded Turing
machine $M$. 
Let $q_0$ be the initial state of $M$ and $q_f$ the unique accepting
state. Let $\Box$ be the blank symbol. 
We can assume that $M$ is non-looping in the sense that there does not exist a
  configuration $c$ such that $c \vdash_M^+ c$. This property can be
  ensured by adding a binary  counter to
  $M$ that is decremented during each transition of the original machine. 
Moreover,  we can assume that
every accepting computation path of $M$ has odd length (i.e.,
  an odd number of transitions), and that every tape cell contains 
$\Box$ as soon as $M$ enters the accepting state $q_f$.
Let $p(n)$ (a polynomial) be the space bound of $M$ and
let $x$ be an input for $M$ of length $n$. Moreover, let $m = p(n)$
and $k = k_M(m+1)$.
By Lemma~\ref{reach} we can compute in space $O(\log m) = O(\log n)$ a layered DFA $A(M,m)$
of depth $k$  
such that
$$
L(A(M,m)) = \{ f_M(c_1) \otimes f_M(c_2) \mid c_1, c_2  \in \Gamma^*Q\Gamma^*,
|c_1| = |c_2|=m+1, c_1 \vdash_M c_2 \}.
$$
Let $w_0 = f_M(q_0 x\Box^{m-n})$ (resp., $w_f = f_M(q_f \Box^m)$)
be the encoding of the initial (resp., accepting) configuration.
Recall that we assume that $0^k$ does not belong to $f_M(\Gamma^*Q\Gamma^*)$.
By taking the direct product of $A(M,m)$
with a layered DFA for the language
$$
K = \{ 0^k \otimes w_0, w_f \otimes 0^k\} 
\cup \{ w \otimes w \mid w \in \{0,1\}^k \setminus 0^* \}
$$
(which can be  computed in space $\log m$), we  obtain a 
layered DFA $A(M,x)$ with
$L(A(M,x)) = L(A(M,m)) \cup K$.
Let $\mathcal{G}(M,x)$ be the directed  graph $\mathcal{G}(A(M,x))$
defined by the DFA $A(M,x)$. Its node set is
$\{0,1\}^k$ and there is an edge from $v$ to $w$ if and only if 
$v \otimes w \in L(A(M,x))$. Let $\text{adj}(M,x)$ be the adjacency
matrix of $\mathcal{G}(M,x)$.
We compute $\det(\text{adj}(M,x))$ by considering 
cycle covers of the graph $\mathcal{G}(M,x)$. 
Note that node $0^k$ lies on a directed cycle if and only if 
there is a path from $w_0$ to $w_f$ in $\mathcal{G}(M,x)$.
Moreover, since $M$ is non-looping, every cycle cover of $\mathcal{G}(M,x)$
consists of a path from $w_0$ to $w_f$ together with the two edges
$(w_f, 0^k)$ and $(0^k, w_0)$ (such a cycle has odd
length and hence is a product of an even number of transpositions) together with loops on the remaining nodes. 
It follows that $\det(\text{adj}(M,x))$ is equal to the number of 
paths from $w_0$ to $w_f$ in $\mathcal{G}(M,x)$.
But this number is equal to the number of accepting computations
of the machine $M$ on input $x$, which is either $0$ or $1$ (since $M$ is deterministic).
By  Lemma~\ref{lemma-layered} applied to the DFA $A(M,x)$, we obtain
in logspace an MTDD $G$ (with integer entries $0$ and $1$ only) 
such that $\val(G) = \text{adj}(M,x)$. This shows the lower bound in (1).

Let us finally prove the lower bound in (2). 
Let us take two polynomial space bounded Turing
machines $M_1$ and $M_2$ with the same input alphabet. 
We can also assume that $M_1$ and $M_2$ have the same
state set $Q$ and tape alphabet $\Gamma$. In particular, we can assume that
$k_{M_1} = k_{M_2}$. Let $f = f_{M_1} = f_{M_2}$ be the binary coding
mapping for $Q \cup \Gamma$.
Let $q_0$ be the initial state of $M_1$ and $M_2$ and $q_f$ the unique accepting
state of $M_1$ and $M_2$. We make the same assumptions that we have made for $M$
in the lower bound proof for statement (1). We can also assume that the polynomial
$p(n)$ is a space bound for $M_1$ as well as $M_2$. 

Let $x$ be an input for $M_1$ and $M_2$ of length $n$,
and let $m = p(n)$, $k =k_{M_1}(m+1) = k_{M_2}(m+1)$. 
With Lemma~\ref{reach} we can construct in space $O(\log m) = O(\log n)$ layered DFAs $A(M_1,m)$
and $A(M_2,m)$ of depth $k$ such that
$$
L(A(M_i,m)) = \{ f(c_1) \otimes f(c_2) \mid c_1, c_2  \in \Gamma^*Q\Gamma^*,
|c_1| = |c_2|=m+1, c_1 \vdash_{M_i} c_2 \}.
$$
Let $w_0 = f(q_0 x \Box^{m-n})$ be the encoding of the initial configuration,
and let $w_f = f(q_f \Box^m)$ be the encoding of the unique accepting configuration.
Recall that we assumed that $0^k$ does not belong to $f(\Gamma^*Q\Gamma^*)$.

From the layered DFAs $A(M_1,m)$
and $A(M_2,m)$ we now construct a layered DFA $A(M_1,M_2,m)$
of depth $k+1$ such that
\begin{eqnarray*}
L(A(M_1,M_2,m)) & = & 
 \{ 0u \otimes 0v \mid u \otimes v \in L(A(M_1,m)) \}\; \cup \; \\ & &  \{ 1u \otimes 1v \mid u \otimes v \in L(A(M_2,m)) \}.
\end{eqnarray*}
For this we basically have to take the disjoint union of $A(M_1,m)$
and $A(M_2,m)$.
By taking the product of $A(M_1,M_2,m)$
with a layered DFA for the language
\begin{align*}
K = & \ \{ 0^{k+1} \otimes 0w_0, 0w_f \otimes 0^{k+1}, 0^{k+1} \otimes 1w_0,
1w_f \otimes 10^k, 10^k \otimes 0^k \}\
\cup \\
& \ \{ w \otimes w \mid w \in \{0,1\}^{k+1} \setminus 0^* \}
\end{align*}
(which can be easily constructed in space $O(\log k) = O(\log n)$), we can obtain a 
layered DFA $A(M_1,M_2,x)$ with
$$
L(A(M_1,M_2,x)) =  L(A(M_1,M_2,m)) \cup K .
$$ 
Let $\mathcal{G}(M_1,M_2,x)$ be the directed  graph $\mathcal{G}(A(M_1,M_2,x))$
defined by the layered DFA $A(M_1,M_2,x)$. This graph consists of the disjoint
union of the two graphs $\mathcal{G}(M_1,m) := \mathcal{G}(A(M_1,m))$ and $\mathcal{G}(M_2,m) := \mathcal{G}(A(M_2,m))$
(basically the configurations graphs of $M_1$ and $M_2$ on
configurations of tape length $m$) together with two nodes
$0^{k+1}$ and $10^k$ and the following edges:
\begin{itemize}
\item Edges from $0^{k+1}$ to $0w_0$ and $1w_0$
  (the copies of the initial
  configuration in the graphs $\mathcal{G}(M_1,m)$ and $\mathcal{G}(M_2,m)$).
\item  An edge from  $0 w_f$ (the copy of the accepting configuration
 in $\mathcal{G}(M_1,m)$) back to $0^{k+1}$.
\item An edge from $1w_f$ (the copy of the accepting configuration
 in $\mathcal{G}(M_2,m)$) to $10^{k}$.
\item An edge from $10^k$ back to $0^{k+1}$.
\item Loops at all nodes except for $0^{k+1}$.
\end{itemize}
Let $\text{adj}(M_1,M_2,x)$ be the adjacency
matrix of the directed graph $\mathcal{G}(M_1,M_2,x)$.
Let us compute $\det(\text{adj}(M_1,M_2,x))$ by considering 
cycle covers of the graph $\mathcal{G}(M_1,M_2,x)$. 
Note that node $0^{k+1}$ lies on a directed cycle if and only if 
there is a path from $w_0$ to $w_f$ in $\mathcal{G}(M_1,x)$
or from $w_0$ to $w_f$ in $\mathcal{G}(M_2,x)$.
Moreover, since $M$ is non-looping, every cycle cover of $\mathcal{G}(M,x)$
consists of loops together with either
\begin{itemize}
\item
a path from $0w_0$ to $0w_f$ (in $\mathcal{G}(M_1,m)$) together with the 
two edge $(0^{k+1},0w_0)$ and $(0w_f, 0^{k+1})$ (every such cycle has odd
length, and hence is a product of an even number of transpositions), or
\item
a path from $1w_0$ to $1w_f$ (in $\mathcal{G}(M_2,m)$) together with the 
three edges $(0^{k+1},1w_0)$, $(1w_f, 10^{k})$, and $(10^{k}, 0^{k+1})$
(every such cycle has even
length, and hence is a product of an odd number of transpositions).
\end{itemize}
It follows that $\det(\text{adj}(M_1,M_2,x))$ is equal to the number of 
paths from $0w_0$ to $0w_f$ in $\mathcal{G}(M_1,m)$
minus 
the number of 
paths from $1w_0$ to $1w_f$ in $\mathcal{G}(M_2,m)$.
But this number is equal to the number of accepting computations
of the machine $M_1$ on input $x$ minus
the number of accepting computations
of the machine $M_2$ on input $x$.
\qed
\end{proof}
Note that the determinant of a diagonal matrix is zero if and only if
there is a zero-entry on the diagonal. This can be easily checked
in polynomial time for a diagonal matrix produced by an MTDD.
For  $\MTDD_+$ (actually, for a sum of several MTDD-represented matrices)
we can show $\NP$-completeness of this problem:

\begin{theorem}\label{thm:MTDD-eq-0-npcomplete}
It is $\NP$-complete to check $\det(\val(G_1)+\cdots+\val(G_k)) = 0$ 
for given MTDDs $G_1, \ldots, G_k$ 
that produce diagonal matrices of the same dimension.
\end{theorem}

\begin{proof}
Membership in $\NP$ is easy: Simply guess a position $1 \leq i \leq 2^n$, compute the values
$n_j = \val(G_j)_{i,i}$ for $1 \leq j \leq k$ and check whether $n_1 + \cdots + n_k = 0$.

Our $\NP$-hardness proof uses again the 3SAT encoding from \cite{berman-karpinski-2d:02}
that we applied in the proof of Theorem~\ref{thm-identity-coNP}.
Take a boolean formula
$C = \bigwedge_{i=1}^m C_i$, where every $C_i$ is a disjunction of three literals.   
Assume that $x_1, \ldots, x_n$ are the boolean variables that occur in $C$.
For each $1 \leq i \leq m$ let $w_i \in \{0,1\}^{2^n}$ be the binary string of length 
$2^n$, where the $j$-th symbol of $w_i$ ($1 \leq k \leq 2^n$) is $1$ if and only if 
the lexicographically $j$-th truth assignment to the variables $x_1, \ldots, x_n$
satisfies clause $C_i$. In     \cite{berman-karpinski-2d:02}  
it is shown that a fully balanced SLP (i.e., an SLP with
a fully balanced derivation tree) for $w_i$ can be constructed in logspace from the clause
$C_i$. We can use the same construction in order to construct in logspace an
MTDD $G_i$ of height $n$ such that $\val(G_i)$ is a diagonal matrix with the word
$w_i$ on the diagonal. Here is the construction: Let $C_i = (\alpha_{j_1} \vee \alpha_{j_2} \vee \alpha_{j_3})$,
where $1 \leq j_1 < j_2 < j_3 \leq n$ and every $\alpha_{j_k}$ is either $x_{j_k}$ or 
$\neg x_{j_k}$.  We take variables $A_0,\ldots, A_n$, $B_0, \ldots, B_{n-1}$, $Z_0, \ldots, Z_{n-1}$,
where $B_i$ produces the $(2^i \times 2^i)$-dimensional identity matrix $I_{2^i}$ and $Z_i$ produces the $(2^i \times 2^i)$-dimensional
zero matrix. For the variables $A_0,\ldots, A_n$ we add the following rules:
For every $1 \leq j \leq n$ with $j \not\in \{j_1, j_2, j_3\}$ take the rule
$$
A_j \to \left( \begin{array}{cc} A_{j-1} \ & \ Z_{j-1} \\[1mm] Z_{j-1} & A_{j-1} \end{array} \right) .
$$
For every  $j \in \{j_1, j_2, j_3\}$ such that $\alpha_j = x_j$ 
take the rule
$$
A_j \to \left( \begin{array}{cc} A_{j-1} \ & \ Z_{j-1} \\[1mm] Z_{j-1} & B_{j-1} \end{array} \right) .
$$
For every  $j \in \{j_1, j_2, j_3\}$ such that $\alpha_j =  \neg x_j$ 
take the rule
$$
A_j \to \left( \begin{array}{cc} B_{j-1} \ & \ Z_{j-1} \\[1mm] Z_{j-1} & A_{j-1} \end{array} \right) .
$$
Finally we take the rule
$A_0 \to 0$. Let $A_n$ be the initial variable of $G_i$. Then, indeed, $\val(G_i)$ 
is a diagonal matrix with the word
$w_i$ on the diagonal for $1 \leq i \leq m$. Let $G_{m+1}$ be an MTDD such that
$\val(G_{m+1}) = -m I_{2^n}$. Then $\val(G_1)+\cdots+\val(G_{m+1})$ is a diagonal
matrix which has a zero on the diagonal (i.e., $\det(\val(G_1)+\cdots+\val(G_{m+1}))=0$)
if and only if the 3CNF formula $C$ is satisfiable.
\qed
\end{proof}

 \subsection{Hardness of iterated multiplication and powering for MTDDs}   \label{sec-powers-hard}
 
Let us now discuss the complexity of iterated multiplication and powering.
Computing a specific entry, say at position $(1,1)$,
of the product of $n$ explicitly given matrices over $\mathbb{Z}$
(resp., $\mathbb{N}$)
is known to be complete for $\mathsf{GapL}$
(resp., $\#\mathsf{L}$) \cite{Toda91countingproblems}. 
Corresponding results hold for the computation of the $(1,1)$-entry
of a matrix power $A^n$, where $n$ is given in unary notation.
Hence, the binary encodings of these numbers can be computed in $\FSPACE(\log^2(n))$.
As usual, these problems become exponentially harder for 
matrices that are encoded by boolean circuits (see Section~\ref{sec-bool-circuits}). 
Let us briefly discuss two scenarios (recall the matrices $M_{C,1}$ and $M_{C,2}$  
defined from a circuit in Section~\ref{sec-bool-circuits}).

\begin{definition} \label{def3}
For a tuple $\overline{C}= (C_1, \ldots, C_n)$ of boolean circuits 
we can define the matrix product
$M_{\overline{C}} = \prod_{i=1}^n M_{C_i,1}$.
\end{definition}

\begin{lemma}   \label{lemma-def3}
The function $\overline{C} \mapsto (M_{\overline{C}})_{1,1}$, where every matrix 
$M_{C_i,1}$ is over $\NN$ (resp., $\ZZ$), belongs to $\mathsf{\#P}$ (resp., $\mathsf{GapP})$.
\end{lemma}

\begin{proof}
Let us first show the result for $\mathsf{\#P}$.
Let $M_{C_i,1} = (a^{(i)}_{j,k})_{1 \leq j,k \leq 2^m}$, where $m=|\overline{x}|=|\overline{y}|$. We have
\begin{equation} \label{multiple-sum}
\big(\prod_{i=1}^n M_i \big)_{1,1} = \sum_{i_1=1}^{2^m} \sum_{i_2=1}^{2^m} \cdots \sum_{i_{n-1}=1}^{2^m}
a^{(1)}_{1,i_1} a^{(2)}_{i_1,i_2} \cdots a^{(n-1)}_{i_{n-2},i_{n-1}} a^{(n)}_{i_{n-1},1} .
\end{equation}
We have to come up with
a nondeterministic polynomial time Turing machine $M$ that has that
many accepting computation paths on  input $(C_1,\ldots, C_n)$. Using
$(n-1) \cdot m$ binary branchings, the machine $M$ can produce an
arbitrary tuple $(i_1,\ldots,i_{n-1})$, where the numbers 
$1 \leq i_1,\ldots, i_{n-1} \leq 2^m$ are written down in binary
notation. Next, 
we can compute in deterministic polynomial time 
the binary codings of all natural numbers 
$a^{(1)}_{1,i_1}, a^{(2)}_{i_1,i_2}, \ldots, a^{(n-1)}_{i_{n-2},i_{n-1}}, a^{(n)}_{i_{n-1},1}$. 
Then we compute the product
$a$ of these numbers again deterministically in polynomial time.
If $a=0$ then we reject on the current computation path (this corresponds to a $0$ 
in the multiple sum \eqref{multiple-sum}). Otherwise,
using the binary coding of $a > 0$ the machine branches $\lceil \log a
\rceil$ many times in order to produce $a$ many accepting computation paths.

For the statement concerning $\mathsf{GapP}$ one can argue similarly.
We have to come up with two polynomial space bounded machines
such that $\big(\prod_{i=1}^{m} M_i \big)_{1,1}$
is equal to the number of accepting computations of the first
machine minus the number of accepting computations of the second
machine. These two machines work as above, but the first (resp. second) machine 
only produces 
$a = a^{(1)}_{1,i_1}, a^{(2)}_{i_1,i_2}, \ldots, a^{(n-1)}_{i_{n-2},i_{n-1}}, a^{(n)}_{i_{n-1},1}$ 
many accepting computation paths if $a > 0$ (resp. $a < 0$).
\qed
\end{proof}

\begin{definition} \label{def4}
A boolean circuit $C(\overline{w},\overline{x},\overline{y},\overline{z})$
with $k = |\overline{w}|$,
$m = |\overline{x}|$, and 
$n=|\overline{y}|=|\overline{z}|$ encodes a  sequence of $2^k$ many 
$(2^n \times 2^n)$-matrices: For every bit vector $\overline{a} \in \{0,1\}^k$, define
the circuit $C_{\overline{a}} = C(\overline{a},\overline{x},\overline{y},\overline{z})$
and the matrix $M_{\overline{a}} = M_{C_{\overline{a}},2}$.
Finally, let $M_C = \prod_{\overline{a} \in \{0,1\}^k} M_{\overline{a}}$ be the product of all these matrices. 
\end{definition}

\begin{lemma} \label{lemma-def4}
The function $C(\overline{w},\overline{x},\overline{y},\overline{z}) \mapsto M_C$
belongs to $\FPSPACE$.
\end{lemma}

\begin{proof}
The lemma follows from Lemma~\ref{PSPACE-counting} and the following two facts:
(i) From the circuit $C(\overline{w},\overline{x},\overline{y},\overline{z})$ one can compute
the tuple of matrices $(M_{C_{\overline{a}},2})_{\overline{a} \in \{0,1\}^k}$
in polynomial space (simply iterate over all valuations for the boolean variables $\overline{w},\overline{x},\overline{y},\overline{z}$),
and (ii) computing an iterated matrix product of explicitly given matrices can be done in $\FSPACE(\log^2(n))$.
\qed
\end{proof}
Lemmas~\ref{lemma-def3} and \ref{lemma-def4} 
yield the upper complexity bounds in the following theorem.

\begin{theorem}\label{thm:matrix-powering}
The following holds: 
\begin{enumerate}[(1)]
\item
The function $(G,n) \mapsto (\val(G)^n)_{1,1}$ with
$G$ an MTDD 
over $\mathbb{N}$ (resp. $\ZZ$)
and $n$  a unary encoded number is
complete for $\mathsf{\#P}$ (resp., $\mathsf{GapP})$.
\item
The function $(G,n) \mapsto (\val(G)^n)_{1,1}$ with $G$ an MTDD 
over $\mathbb{N}$ (resp. $\ZZ$)
and $n$  a binary encoded number
is $\mathsf{\#PSPACE}$-complete  (resp., $\mathsf{GapPSPACE}$-complete).
\end{enumerate}
\end{theorem}

\begin{proof}
It remains to prove the lower bound, for which
we use again succinct versions of Toda's techniques from \cite{Toda91countingproblems},
similar to the proof of Theorem~\ref{theorem-det}.

Let us start with the statements concerning 
$\mathsf{\#P}$ and $\mathsf{\#PSPACE}$.
We start with (1).
Let $M$ be a fixed nondeterministic polynomial time Turing machine. 
One can assume that all maximal computations
of $M$ on an input $x$ of length $n$ have length $p(n)$ for some polynomial $p$.
Let $x$ be an input for $M$ of length $n$, and let 
$m = p(n)$ and $k = k_M(m+1)$. 
We now apply the construction from the 
proof of Lemma~\ref{reach} to $M$ and $m$. 
We obtain a layered DFA $A(M,m)$
such that
$$
L(A(M,m)) = \{ f_M(c_1) \otimes f_M(c_2) \mid c_1, c_2  \in \Gamma^*Q\Gamma^*,
|c_1| = |c_2|=m+1, c_1 \vdash_M c_2 \}.
$$
Let $w_0 = f_M(q_0 x \Box^{m-n})$ be the encoding of the initial configuration,
and $w_f = f_M(q_f \Box^m)$ be the encoding of the unique accepting configuration.
Recall that $0^k$ does not belong to $f_M(\Gamma^*Q\Gamma^*)$.
As in the proof of Theorem~\ref{theorem-det} we obtain
a layered DFA $A(M,x)$
such that
$$
L(A(M,x)) = L(A(M,m))  \cup \{ 0^k \otimes w_0, w_f \otimes 0^k \}.
$$
Let $\mathcal{G}(M,x)$ be the directed  graph  $\mathcal{G}(A(M,x))$, whose node set is
$\{0,1\}^k$ and there is an edge from $v$ to $w$ if and only if 
$v \otimes w \in L(A(M,x))$. Let $\text{adj}(M,x)$ be the adjacency
matrix of $\mathcal{G}(M,x)$.
As in the proof of Theorem~\ref{theorem-det} we obtain an MTDD $G$
such that $\val(G) = \text{adj}(M,x)$.

Then the number of accepting computations of the machine $M$ on input
$x$ is equal to the number of paths of length $p(n)+2$ in the graph 
$\mathcal{G}(M,x)$ from node $0^k$ to node $0^k$. This number is
equal to $(\val(G)^{p(n)+2})_{1,1}$.

The  $\mathsf{\#PSPACE}$-hardness in point (2) of the theorem is proven in the
same way. For a nondeterministic polynomial space bounded Turing-machine
one can assume that all maximal computations
of $M$ on an input $x$ of length $n$ have length $2^{p(n)}$ for some
polynomial $p$.
Hence, we only have to replace the number $m+2$ in the above proof
by $2^m+2$.

Let us now turn to the lower bounds concerning $\mathsf{GapP}$ and $\mathsf{GapPSPACE}$
in the theorem. The proofs are very
similar to the corresponding proofs for $\mathsf{\#P}$ and
$\mathsf{\#PSPACE}$, respectively.
We only consider (2). 
We have to come up with an MTDD over $\{0,1-1\}$. 
Such an MTDD corresponds to a layered DFA, where the last
layer contains three states, corresponding to the three possible 
matrix entries $0$, $1$, and $-1$. Now, take two polynomial space
bounded Turing machines $M_1$ and $M_2$ (with the same input
alphabet), such that all accepting computations of $M_1$ and $M_2$
on an input of length $m$ have length $2^{p(m)}$. Moreover, let $x$
be an input for $M_1$ and $M_2$. 
We have to come up with a layered DFA (with three nodes in the 
last layer) that defines the following $\{1,-1\}$-labeled directed
graph $\mathcal{G}$:
\begin{itemize}
\item $\mathcal{G}$ consists of a disjoint copy of
  $\mathcal{G}(M_1,m)$ and $\mathcal{G}(M_2,m)$ (all edges are
  labelled with $1$) together with an additional node $s$.
\item There is a $1$-labeled edge from node $s$ to the copy of the
  initial configuration of $M_1$ in $\mathcal{G}(M_1,m)$.
\item There is a $-1$-labeled edge from node $s$ to the copy of the
  initial configuration of $M_2$ in $\mathcal{G}(M_2,m)$.
\item There are $1$-labeled edges from the copies of the unique
  accepting configurations in $M_1$ and $M_2$, respectively, back
  to node $s$.
\end{itemize}
Analogously to the construction in the proof of (2) from 
Theorem~\ref{theorem-det} we can construct such a layered DFA.
For the MTDD $G$ over $\{0,1,-1\}$ corresponding to this layered
DFA, $(\val(G)^{2^{p(m)}+2})_{1,1}$ is equal to the number of 
accepting computations of $M_1$ on input $x$ minus the number of 
accepting computations of $M_2$ on input $x$. 
\qed
\end{proof}
By Theorem~\ref{thm:matrix-powering}, there is no polynomial time algorithm that
computes for a given MTDD $G$ and a unary number $n$
a boolean circuit (or even an $\MTDD_+$) for the power $\val(G)^n$, unless $\mathsf{\#P} = \mathsf{FP}$.

By \cite{Toda91countingproblems} and Theorem~\ref{thm:matrix-powering},
the complexity of computing a specific entry of a matrix power $A^n$ covers three different
counting classes, depending on the representation of the matrix $A$ and the exponent $n$ 
(let us assume that $A$ is a matrix over $\mathbb{N}$):
\begin{itemize}
\item $\#\mathsf{L}$-complete, if $A$ is given explicitly and $n$ is given unary.
\item $\#\mathsf{P}$-complete, if $A$ is given by an MTDD and $n$ is given unary.
\item $\#\mathsf{PSPACE}$-complete, if $A$ is given by an MTDD and $n$ is given binary.
\end{itemize}
Let us also mention that in \cite{CMcKTV98,GalotaV05,MereghettiP00} the complexity of evaluating iterated matrix products 
and matrix powers in a fixed
dimension is studied. 
It turns out that multiplying a sequence  of $(d \times d)$-matrices over $\mathbb{Z}$ 
in the fixed dimension $d \geq 3$
is complete for the class $\mathsf{GapNC}^1$ (the counting version of the circuit
complexity class $\mathsf{NC}^1$) \cite{CMcKTV98}. It is open whether the same problem
for matrices over $\mathbb{N}$ is complete for $\#\mathsf{NC}^1$. Moreover, the case $d=2$ is open
too.  Matrix powers for matrices in a fixed dimension can be computed in $\mathsf{TC}^0$ (if the exponent
is represented in unary notation) using the 
Cayley-Hamilton theorem \cite{MereghettiP00}.
Finally, multiplying a sequence of $(d \times d)$-matrices
that is given succinctly by a boolean circuit captures the class $\FPSPACE$ for any $d \geq 3$ \cite{GalotaV05}.

For the problem, whether a power of an MTDD-encoded matrix is zero
(a variant of the classical mortality problem) we can finally show the following:

\begin{theorem}\label{thm:matrix-multi-coNP}
It is $\coNP$-complete (resp.,$\PSPACE$-complete) to check 
whether $\val(G)^m$ is the zero matrix for a given 
MTDD $G$ and a unary (resp., binary) encoded
number $m$.
\end{theorem}

\begin{proof}
Take the construction from the proof of the lower bound from point (1) of
Theorem~\ref{thm:matrix-powering}. 
Recall that $p(n)$ was the time bound of $M$. We assumed that
all maximal computation paths for an input of length $n$ have length
exactly $p(n)$. Let $m=p(n)$.
We can modify the Turing machine $M$ in such a way that 
the graph $\mathcal{G}(M,m)$ (the configuration graph of $M$ on configurations
of tape length $m$)
 does not have directed paths of 
length larger than $m$ (e.g. by splitting the tape of $M$ into 
two tracks and incrementing a unary counter on the second track).
This means that in the graph $\mathcal{G}(M,x)$ there is a path of
length $m+2$ if and only if $x$ is accepted by $M$.
Thus, $x$ is accepted by $M$ if and only if 
$\val(G)^{p(n)+2}$ is not the zero matrix.
The statement concerning $\PSPACE$-completeness is proven in the
same way (we just have to ensure by adding a binary counter on the second track that 
the graph $\mathcal{G}(M,m)$ does not have directed paths of 
length larger than $2^{p(n)}$).
\qed
\end{proof}
Here is a more direct proof for the $\coNP$-hardness statement in Theorem~\ref{thm:matrix-multi-coNP},
which uses a reduction from the complement of 3SAT.

\medskip
\noindent
{\em Alternative proof of Theorem~\ref{thm:matrix-multi-coNP}.}
Let $C = \bigwedge_{i=1}^m C_i$ be a 3CNF formula. In the proof of
Theorem~\ref{thm:MTDD-eq-0-npcomplete} we constructed MTDD 
$G_1, \ldots, G_m$ such that $\val(G_i)$ is the diagonal matrix, 
where the diagonal is the binary string of all truth values
of the clause $C_i$, taken in lexicographic order.
From the MTDD 
$G_1, \ldots, G_m$ we easily obtain an MTDD $G$ such that
$$
\val(G) = \left(\begin{array}{ccccccc}
0 & \val(G_1) & 0 & 0 & \cdots & 0 & 0  \\
0 & 0 & \val(G_2) & 0 & \cdots & 0 & 0  \\
0 & 0 & 0 & \val(G_3) & \cdots & 0 & 0  \\
  &   &   & \vdots    &        &   & \\
0 & 0 & 0 & 0 & \cdots & \val(G_{m-1}) & 0 \\
0 & 0 & 0 & 0 & \cdots & 0 & \val(G_m) \\
0 & 0 & 0 & 0 & \cdots & 0 & 0 
\end{array}\right).
$$
Here, we have to assume that $m+1$ is a power of two, which can 
be enforced by adding dummy clauses.
Since the matrices $\val(G_i)$ commute (they are diagonal matrices)
and are idempotent (since all diagonal values are $0$ or $1$),
the matrix $\val(G)^m$ contains only $0$-blocks except for the
top right-most block, which is $\prod_{i=1}^m \val(G_i)$.
Thus, $\val(G)^m$ is the zero matrix if and only if $C$ is
unsatisfiable.
\qed

\section{Conclusion and future work}
   
   We studied algorithmic problems on matrices that are given by multi-terminal decision diagrams
   enriched by the operation of matrix addition. Several  important
   matrix problems can be solved in polynomial time for this representation, e.g., equality checking, computing 
   matrix entries, matrix multiplication, computing the trace, etc. 
   On the other hand, computing determinants,  matrix powers, and iterated matrix products are
   computationally hard.  For further research, it should be investigated whether the polynomial time problems, like equality test,
   belong to $\NC$.
   Also an experimental implementation is planned for testing practical efficiency.
  

\end{document}